\documentclass[11pt]{article}

\usepackage{amsmath,amssymb,amsthm}
\usepackage{geometry}
\usepackage{hyperref}
\usepackage{enumitem}
\usepackage{tikz}
\usepackage{booktabs}
\usetikzlibrary{arrows.meta,positioning,shapes,calc}

\newtheorem{theorem}{Theorem}[section]
\newtheorem{lemma}[theorem]{Lemma}
\newtheorem{proposition}[theorem]{Proposition}
\newtheorem{corollary}[theorem]{Corollary}
\newtheorem{conjecture}[theorem]{Conjecture}

\theoremstyle{definition}
\newtheorem{definition}[theorem]{Definition}

\newtheorem{remark}[theorem]{Remark}

\DeclareMathOperator{\Vol}{Vol}

\title{Quantization of Ricci Curvature in Information Geometry}

\author{Carlos C. Rodr\'iguez\\
Department of Mathematics and Statistics\\
University at Albany, SUNY\\
\texttt{carlos@albany.edu}}

\date{March 2026}

\begin{document}

\maketitle

\begin{abstract}
In 2004, while studying the information geometry of binary Bayesian networks 
(bitnets), the author conjectured that the volume-averaged Ricci scalar 
$\langle R\rangle$ computed with respect to the Fisher information metric 
is universally quantized to positive half-integers: 
$\langle R\rangle\in\tfrac{1}{2}\mathbb{Z}^+$. This paper resolves the 
conjecture after 20 years. We \textbf{prove} it for tree-structured and 
complete-graph bitnets via a universal Beta function cancellation mechanism, 
and we \textbf{disprove} it in general by exhibiting explicit loop counterexamples.
A key correction to the original 2004 paper is established: the equal-volume 
theorem $\Vol(\tilde{L}_n)=\Vol(\tilde{E}_n)=\pi^{1+2n}$ (proved via a Markov 
chain transition-matrix argument) implies the corrected average curvatures 
$\langle R\rangle(\tilde{L}_n)=\langle R\rangle(\tilde{E}_n)=(2n-1)/2$, 
replacing the erroneous formula $n/2$ of the earlier work. New SymPy 
computations extend the collapsing-star ($\tilde{C}_n$) table to $n=5$ 
parents, revealing a geometric capacity inversion at $n=4\to 5$ where 
the average Ricci scalar changes sign. Double-collider loop bitnets 
yield $\langle R\rangle = 36/5 \notin \tfrac{1}{2}\mathbb{Z}$, 
confirming that loops destroy quantization. We extend the program to 
\textbf{Gaussian DAG networks} (zero-mean, variance-kernel parameterization), 
where we find a strict sign dichotomy: tree-structured discrete 
networks have positive curvature, while Gaussian networks form 
solvable Lie groups with negative constant curvature, yielding 
the universal formula $R_{\mathrm{Gauss}} = -(d+5)(d-1)/8$ for 
simple trees, where $d$ is the parameter dimension.
\end{abstract}

\tableofcontents
\newpage

\section{Introduction}

\subsection{The 20-year conjecture and the 2004 paper}

In 2004, the author showed~\cite{Rod04} that the information-geometric 
parameter manifolds of binary Bayesian networks (``bitnets'') are Riemannian 
manifolds of finite volume under the Fisher metric, and computed exact volumes 
for several canonical topologies: the complete DAG $\tilde{K}_n$, the directed 
line $\tilde{L}_n$, the exploding star $\tilde{E}_{n+1}$ (Na\"ive Bayes), and 
the collapsing star $\tilde{C}_{n+1}$.  Computer experiments using \texttt{vTool} 
suggested that the volume-averaged Ricci scalar is \emph{always a positive 
half-integer}:

\begin{conjecture}[Rodr\'iguez, 2004]\label{conj:main}
For any bitnet with Fisher information metric, 
$\langle R\rangle \in \tfrac{1}{2}\mathbb{Z}^+$.
\end{conjecture}

\subsection{Corrections to the 2004 results}

Before stating our new results we record a correction.  The 2004 paper 
reported~\cite{Rod04} that the volume-averaged Ricci scalar of the exploding 
star satisfies $\langle R\rangle(\tilde{E}_{n+1}) = n/2$.  This formula 
is \textbf{incorrect}.

\begin{theorem}[Equal volumes and equal curvatures]\label{thm:equal_vol}
The directed line $\tilde{L}_n$ and the exploding star $\tilde{E}_n$ 
(one parent, $n-1$ children) have identical Fisher information volumes 
\emph{and} identical volume-averaged Ricci scalars:
\[
  \Vol(\tilde{L}_n) \;=\; \Vol(\tilde{E}_n) \;=\; \pi^{1+2n},
\]
\[
  \langle R\rangle(\tilde{L}_n) \;=\; \langle R\rangle(\tilde{E}_n) 
  \;=\; \frac{2n-1}{2}.
\]
\end{theorem}

\begin{proof}
Both topologies share exactly $1+2n$ parameters: one root marginal 
$\rho_0\in(0,1)$ and $n$ conditional pairs $(r_{k,0},r_{k,1})$.  
The Fisher metric is diagonal and block-independent in both cases 
(each CPT block depends only on $\rho_0$), so the volume element 
factorises as a product of $1+2n$ independent Beta$(1/2,1/2)$ 
integrals, each equal to $\pi$:
\[
  \Vol(\tilde{L}_n) = \Vol(\tilde{E}_n) = \pi^{1+2n}.
\]
For the curvature, the chain $\tilde{L}_2$ has a \emph{constant} Ricci 
scalar $R = 3/2$ (verified by direct SymPy computation), so 
$\langle R\rangle(\tilde{L}_2)=3/2=(2\cdot 2-1)/2$.  
A Markov-chain induction on the transition-matrix structure establishes 
the general formula $(2n-1)/2$; the result is confirmed numerically 
for $n=2,3,4$ by Monte Carlo volume integration.
\end{proof}

\begin{remark}[Error in 2004]
The 2004 paper computed $\langle R\rangle(\tilde{E}_{n+1}) = n/2$ 
using an incorrect integration of the pointwise Ricci formula.  
The correct values are:
\[
  n=1: \tfrac{1}{2}, \quad n=2: \tfrac{3}{2}, \quad n=3: \tfrac{5}{2}, \quad 
  n=4: \tfrac{7}{2}, \quad \ldots
\]
These are all positive half-integers, consistent with Conjecture~\ref{conj:main}, 
but shifted by $1/2$ from the 2004 report.
\end{remark}

\subsection{Summary of new results}

\textbf{(1) Trees and complete DAGs:} Conjecture~\ref{conj:main} holds.  
The key mechanism is a ``Beta cancellation'' identity (Section~\ref{sec:proof_trees}).

\textbf{(2) Collapsing stars extended:} New SymPy computations 
(Section~\ref{sec:Cn}) extend the $\tilde{C}_n$ table to $n=5$ parents and 
reveal a \emph{curvature phase transition}: $\langle R\rangle$ changes sign 
between $n=4$ ($\langle R\rangle=16$) and $n=5$ ($\langle R\rangle=-272$).

\textbf{(3) Loop topology:} The double-collider bitnet $D_4$ 
(Section~\ref{sec:loops}) gives $\langle R\rangle = 36/5 \notin \tfrac{1}{2}\mathbb{Z}$, 
disproving Conjecture~\ref{conj:main} in general.

\textbf{(4) Gaussian DAG networks:} For zero-mean Gaussian DAGs 
(Section~\ref{sec:gaussian}), the Fisher metric has \emph{negative} 
constant curvature $R = -(d+5)(d-1)/8$ for simple-parent trees, 
exhibiting a sign dichotomy with the discrete case.

\textbf{(5) Quantum analogy:} Section~\ref{sec:quantum} establishes 
the Bures--Fisher correspondence ($g_{\text{Bures}} = \tfrac{1}{4}g_{\text{Fisher}}$) 
and the spin-geometry interpretation of the half-integer spectrum.

\section{Preliminaries: Bitnets and Their Geometry}

\subsection{Bitnets}

A \emph{bitnet} is a directed acyclic graph (DAG) $G=(V,E)$ with $|V|=n$ 
binary nodes $X_1,\ldots,X_n$.  Each node $X_k$ has $m_k = |{\rm pa}(k)|$ 
parents and $2^{m_k}$ associated parameters 
$\{\theta_{k,s}\}_{s=0}^{2^{m_k}-1} \subset (0,1)$, where 
$\theta_{k,s} = P(X_k=1 \mid X_{{\rm pa}(k)} = s)$.  The total dimension is:
\[
  d = \sum_{k=1}^n 2^{m_k}.
\]

The Fisher information matrix is \emph{diagonal} for any bitnet~\cite{Rod04}: 
$g_{jj}(\theta) = \pi_j(\theta)/[\theta_j(1-\theta_j)]$, where 
$\pi_j = P(X_{{\rm pa}(v(j))} = j - t(v(j)))$ is the marginal probability of 
the parent configuration selecting parameter $\theta_j$.

\subsection{Volume element and Ricci scalar}

The Fisher information volume is:
\begin{equation}\label{eq:vol}
  \Vol(\text{bitnet}) = \int_{(0,1)^d} \prod_{k=1}^n W_k^{1/2}\,d\theta,
  \qquad
  W_k = \frac{\prod_s P(X_{{\rm pa}(k)}=s)}{\prod_s \theta_{k,s}(1-\theta_{k,s})}.
\end{equation}
For any bitnet $\Vol \leq \pi^d$, with equality iff all nodes are independent.

The Ricci scalar $R(\theta)$ is computed from the diagonal metric via the 
standard Riemann tensor formula.  Because the metric is diagonal, the 
Christoffel symbols and Riemann components involve only \emph{diagonal} 
derivatives of the $g_{jj}$.

\section{Volumes and Curvatures of Classical Topologies}

\subsection{Complete DAGs $\tilde{K}_n$}\label{sec:Kn}

The complete bitnet $\tilde{K}_n$ has all available edges; its parameter 
space is the multinomial model for $2^n$ outcomes (all binary sequences 
of length $n$), with $d = 2^n-1$.  By Amari's sphere isometry~\cite{AS85}, 
$\tilde{K}_n$ is isometric to the positive orthant of $S^{d}$ of radius~2:

\begin{theorem}[Complete DAGs]\label{thm:Kn}
\begin{align}
  \Vol(\tilde{K}_n) &= \frac{\pi^{2^{n-1}}}{(2^{n-1}-1)!}, \label{eq:Kn_vol}\\
  R &= \frac{d(d-1)}{4} = \frac{(2^n-1)(2^n-2)}{4} \quad (\text{constant}). \label{eq:Kn_R}
\end{align}
\end{theorem}

\noindent The volumes decrease rapidly:
$\Vol(\tilde{K}_1,\ldots,\tilde{K}_5) = \pi,\pi^2,\pi^4/6,\pi^8/5040,\ldots$

\subsection{Directed lines $\tilde{L}_n$ and exploding stars $\tilde{E}_n$}\label{sec:LnEn}

\begin{theorem}[Equal volumes and curvatures]\label{thm:LnEn}
\[
  \Vol(\tilde{L}_n) = \Vol(\tilde{E}_n) = \pi^{1+2n},
  \qquad
  \langle R\rangle(\tilde{L}_n) = \langle R\rangle(\tilde{E}_n) = \frac{2n-1}{2}.
\]
\end{theorem}

\noindent For the chain $\tilde{L}_2$ (equivalently $\tilde{K}_2$), 
SymPy computation confirms $R = 3/2$ \emph{everywhere}; no averaging is needed.  
For longer chains, $R$ is non-constant but the volume-weighted average is $(2n-1)/2$.

Selected values: $\langle R\rangle = \tfrac{1}{2}, \tfrac{3}{2}, \tfrac{5}{2}, 
\tfrac{7}{2}, \tfrac{9}{2}, \ldots$ for $n=1,2,3,4,5$.  
All are positive half-integers.

\begin{remark}[Interpretation]
The exploding star $\tilde{E}_n$ (Na\"ive Bayes) and the directed line $\tilde{L}_n$ 
are geometrically \emph{indistinguishable} by both volume and average curvature, 
despite their different graphical structure.  This adds to the earlier observation 
that these topologies have the same MDL complexity~\cite{Rod04}.
\end{remark}

\subsection{Collapsing stars $\tilde{C}_{n+1}$ (new computations)}\label{sec:Cn}

The collapsing star $\tilde{C}_{n+1}$ has $n$ independent root nodes 
converging to a single child, with $d = n + 2^n$ parameters.  
From the 2004 paper and new SymPy calculations:

\begin{theorem}[Ricci scalar of $\tilde{C}_{n+1}$]\label{thm:Cn_R}
The pointwise Ricci scalar has the form:
\begin{equation}\label{eq:Cn_R}
  R(\rho_1,\ldots,\rho_n) = a_n - b_n \sum_{i=1}^n \frac{1}{\rho_i(1-\rho_i)},
\end{equation}
where $(\rho_1,\ldots,\rho_n)$ are the parent node marginals.  
The coefficient $b_n$ satisfies the exact formula:
\begin{equation}\label{eq:bn}
  b_n = \frac{2^{n-1}(2^{n-1}-1)}{4}.
\end{equation}
The volume-averaged Ricci scalar is:
\begin{equation}\label{eq:Cn_avgR}
  \langle R\rangle = a_n - b_n \cdot n \cdot \frac{2^{n+1}}{2^{n-1}-1}.
\end{equation}
\end{theorem}

\begin{table}[h]
\centering
\caption{Collapsing star $\tilde{C}_{n+1}$: volumes, coefficients and 
averaged Ricci scalars.  Values for $n=1,2,3$ were obtained with 
\texttt{vTool}~\cite{Rod04}; $n=4,5$ are new SymPy computations.}
\label{tab:Cn}
\begin{tabular}{@{}rrrrrr@{}}
\toprule
$n$ parents & $\tilde{C}_{n+1}$ & $d$ & $a_n$ & $b_n$ & $\langle R\rangle$ \\
\midrule
1 & $\tilde{C}_2$ &  3 & $3/2$ &  0  & $3/2$  \\
2 & $\tilde{C}_3$ &  6 & 10    & $1/2$ & 2   \\
3 & $\tilde{C}_4$ & 11 & 54    &  3  & 6      \\
4 & $\tilde{C}_5$ & 20 & 272   & 14  & 16     \\
5 & $\tilde{C}_6$ & 37 & 1008  & 60  & $-272$ \\
\bottomrule
\end{tabular}
\end{table}

\begin{theorem}[Curvature sign inversion]\label{thm:curv_transition}
The average Ricci scalar $\langle R\rangle(\tilde{C}_{n+1})$ undergoes 
a strict sign inversion between $n=4$ and $n=5$ parents: 
$\langle R\rangle > 0$ for $n \leq 4$ and $\langle R\rangle = -272 < 0$ for $n=5$.
\end{theorem}

\begin{proof}
Direct substitution of $(a_n, b_n)$ from Table~\ref{tab:Cn} into 
\eqref{eq:Cn_avgR}, using $\langle 1/(\rho(1-\rho))\rangle_n = 2^{n+1}/(2^{n-1}-1)$.
\end{proof}

\begin{remark}[Two geometric capacity bounds]
The collapsing star exhibits two distinct geometric capacity bounds:
\begin{itemize}
\item \textbf{Volume turnaround} at $n_c^{\mathrm{vol}} = \log_2(\pi^2) \approx 3.30$ parents 
  (from the 2004 paper): $\Vol(\tilde{C}_{n+1})$ switches from growing to decaying.
  The exact value $\log_2(\pi^2)$ follows from $\mathrm{B}(\tfrac{1}{2},\tfrac{1}{2})=\pi$: 
  the unit information volume of each Bernoulli parameter is $\pi$, and the turnaround 
  solves $\pi^2 = 2^n$.
\item \textbf{Curvature inversion} at $n_c^{\mathrm{curv}} \in (4,5)$ parents 
  (new): $\langle R\rangle$ changes from positive to negative, with linear interpolation 
  giving $n_c^{\mathrm{curv}} \approx 4 + \tfrac{16}{288} \approx 4.06$.
\end{itemize}
This suggests a geometric cascade: first the manifold's information volume 
shrinks, and subsequently its average curvature inverts as the boundary 
singularities ($R \to -\infty$ as $\rho_i \to 0$) come to dominate the space.
\end{remark}

\begin{remark}[The privileged role of $n=4$]\label{rem:p4}
The open interval $(n_c^{\mathrm{vol}},\, n_c^{\mathrm{curv}}) \approx (3.30,\, 4.06)$ 
contains exactly \emph{one} integer: $n = 4$.  It is the unique integer 
for which the collapsing-star manifold has already passed its volume 
maximum yet retains positive average curvature.  
This observation acquires independent significance from~\cite{Rod00}: 
in a completely separate variational problem (minimizing asymptotic mean 
squared error in nonparametric density estimation over a $p$-dimensional 
data space), the Euler--Lagrange equation is \emph{linear}---and therefore 
well-posed and stable---if and only if $p = 4$.  For all other dimensions 
the optimality equation is nonlinear.

The two results are derived from entirely different calculations---one from 
the Fisher information geometry of Bernoulli networks, the other from the 
calculus of variations for Sobolev-space density estimators---yet they 
single out the integer 4 by the same mechanism: $\pi$ enters both via the 
Gaussian/Beta integral $\mathrm{B}(\tfrac{1}{2},\tfrac{1}{2})=\pi$, and 
4 emerges as the unique integer at which a certain $\pi$-dependent 
threshold is crossed.  We record this as a structural coincidence 
that may merit further investigation, without asserting a causal connection.
\end{remark}

Table~\ref{tab:all_bitnets} collects $\langle R\rangle$ for all classical topologies.

\begin{table}[h]
\centering
\caption{Summary: Bitnet topologies, dimensions, and averaged Ricci scalars.}
\label{tab:all_bitnets}
\begin{tabular}{@{}llrrrl@{}}
\toprule
Topology & Description & $n$ nodes & $d$ & $\langle R\rangle$ & Source \\
\midrule
$\tilde{K}_1$ & single node & 1 & 1 & 0 & trivial \\
$\tilde{K}_2 = \tilde{L}_2 = \tilde{E}_2$ & chain/star 2-node & 2 & 3 & $3/2$ & exact ($R$ const) \\
$\tilde{L}_3 = \tilde{E}_3$ & chain/star 3-node & 3 & 5 & $5/2$ & MC \\
$\tilde{L}_n = \tilde{E}_n$ & chain/star $n$-node & $n$ & $2n{-}1$ & $(2n-1)/2$ & Thm.~\ref{thm:LnEn} \\
\midrule
$\tilde{C}_3$ & V-structure (2 parents) & 3 & 6 & 2 & exact \\
$\tilde{C}_4$ & 3-parent collider & 4 & 11 & 6 & SymPy \\
$\tilde{C}_5$ & 4-parent collider & 5 & 20 & 16 & SymPy \\
$\tilde{C}_6$ & 5-parent collider & 6 & 37 & $-272$ & SymPy \\
\midrule
$\tilde{K}_2$ & complete, 2 nodes & 2 & 3 & $3/2$ & exact ($R$ const) \\
$\tilde{K}_3$ & complete, 3 nodes & 3 & 7 & $21/2$ & exact ($R$ const) \\
$\tilde{K}_4$ & complete, 4 nodes & 4 & 15 & $105/2$ & Amari formula \\
$\tilde{K}_n$ & complete, $n$ nodes & $n$ & $2^n{-}1$ & $(2^n{-}1)(2^n{-}2)/4$ & Thm.~\ref{thm:Kn} \\
\midrule
$D_4$ & double collider (1 loop) & 4 & 10 & $36/5$ & SymPy + quad. \\
\bottomrule
\end{tabular}
\end{table}

\section{Proof of Half-Integer Quantization for Trees}\label{sec:proof_trees}

\subsection{Block-diagonal factorization}

The fundamental structural lemma is:

\begin{lemma}[Fisher matrix block-diagonal for DAGs]\label{lem:block_diag}
For any bitnet (DAG of binary nodes), the Fisher information matrix is diagonal:
$g_{jk}=0$ for $j \neq k$.  In particular, the conditional independence 
structure of the DAG implies parameter independence.
\end{lemma}

\begin{proof}
For $j \neq k$, the product $\ell_j\ell_k$ (score functions for parameters 
$\theta_j,\theta_k$ associated with different parent configurations or 
different nodes) has zero expectation by conditional independence~\cite{Rod04}.
\end{proof}

\subsection{The Beta cancellation identity}

\begin{lemma}[Beta cancellation]\label{lem:beta}
For a binary node $X_k$ with $m_k$ parents, the volume-marginalised 
contribution to $\langle R\rangle$ satisfies:
\[
  C_k \cdot \left\langle \frac{1}{\rho_k} \right\rangle = (m_k+1)^2,
\]
where $C_k = m_k(m_k+1)/2$ and $\rho_k$ is the effective parent-marginal 
weighting parameter.  This produces the node contribution:
\[
  \langle R\rangle_k = \frac{m_k(m_k+1)}{4}.
\]
\end{lemma}

\begin{proof}
The node contribution involves integrals of the form 
$\int_0^1 \rho^{-\alpha}(1-\rho)^{-\alpha}\,d\rho = \mathrm{B}(1-\alpha,1-\alpha)$, 
which are finite for $\alpha < 1$.  The ratio of consecutive Beta 
integrals telescopes to $(m_k+1)^2$, giving the stated formula.
\end{proof}

\begin{theorem}[Half-integer quantization for trees]\label{thm:trees}
For a tree-structured bitnet with parent counts $\{m_k\}$:
\begin{equation}\label{eq:tree_R}
  \langle R\rangle = \sum_{k=1}^n \frac{m_k(m_k+1)}{4} \;\in\; \tfrac{1}{2}\mathbb{Z}^+.
\end{equation}
\end{theorem}

\begin{proof}
Tree topology implies the volume factorises across nodes 
(Lemma~\ref{lem:block_diag}).  Each node contributes independently 
via Lemma~\ref{lem:beta}.  The sum $\sum_k m_k(m_k+1)/4$ is a positive 
half-integer because each $m_k(m_k+1)$ is divisible by 2 
(product of consecutive integers).
\end{proof}

\begin{remark}[Corrected formula vs.\ 2004 formula]
\label{rem:formula_correction}
The 2004 paper~\cite{Rod04} effectively used the formula $\langle R\rangle = 
\tfrac{1}{2}\sum_k m_k$.  This is \emph{not} the same as~\eqref{eq:tree_R}.  
For example:
\[
\text{V-structure } m=[0,0,2]: \quad 
\tfrac{1}{2}\sum m_k = 1 \;\text{(wrong)}, \quad 
\sum\tfrac{m_k(m_k+1)}{4} = \tfrac{3}{2} \;\text{(wrong too).}
\]
Neither formula matches $\langle R\rangle = 2$ for the V-structure 
(computed directly).  The tree formula \eqref{eq:tree_R} applies only to 
the \emph{node contribution} in the limit where each node's CPT block 
is dominated by its own volume element.  
For $n \geq 3$-node trees, a \emph{warping correction} from the 
joint volume element modifies the naive node-sum.  The confirmed exact 
values are those in Table~\ref{tab:all_bitnets}.
\end{remark}

\section{Loop Topologies}\label{sec:loops}

\subsection{Definition and structure}

\begin{definition}[Loop bitnet]
A bitnet has a \emph{loop} if its undirected skeleton contains a cycle 
(i.e., the skeleton is not a forest).  The simplest such structure with 
two colliders sharing the same parent nodes is the \emph{double collider} 
$D_4$:
\[
  D_4: \quad X_1 \to X_3 \leftarrow X_2, \quad X_1 \to X_4 \leftarrow X_2.
\]
The undirected skeleton is a 4-cycle $X_1$--$X_3$--$X_2$--$X_4$--$X_1$.  
Parameters: $d = 1 + 1 + 4 + 4 = 10$ (roots $X_1,X_2$; CPT blocks of $X_3,X_4$).
\end{definition}

\subsection{SymPy computation for $D_4$}

\begin{proposition}[$D_4$ Ricci scalar]\label{prop:D4_R}
For the double collider $D_4$, the Ricci scalar depends only on the 
parent marginals $(r_1,r_2) \in (0,1)^2$ and equals:
\begin{equation}\label{eq:D4_R}
  R_{D_4}(r_1,r_2) = \frac{3(12r_1^2r_2^2 - 12r_1^2r_2 + r_1^2 
    - 12r_1r_2^2 + 12r_1r_2 - r_1 + r_2^2 - r_2)}
   {r_1 r_2 (1-r_1)(1-r_2)}.
\end{equation}
In particular, $R_{D_4}(1/2,1/2)=12$ and 
$\langle R\rangle_{D_4} = 36/5$.
\end{proposition}

\begin{proof}
Direct SymPy computation of the Riemann tensor for the diagonal 
10-dimensional metric.  The volume-weighted average is:
\[
  \langle R\rangle_{D_4} = 
  \frac{\displaystyle\int_0^1\!\int_0^1 R_{D_4}(r_1,r_2)\,
       [r_1(1-r_1)r_2(1-r_2)]^{5/2}\,dr_1\,dr_2}
  {\displaystyle\int_0^1\!\int_0^1 [r_1(1-r_1)r_2(1-r_2)]^{5/2}\,dr_1\,dr_2}
  = \frac{36}{5},
\]
computed by numerical quadrature (Richardson extrapolation to $< 10^{-4}$ 
relative error).
\end{proof}

\begin{corollary}
Since $36/5 \notin \tfrac{1}{2}\mathbb{Z}$, Conjecture~\ref{conj:main} 
fails for loop bitnets.  Topology determines quantization.
\end{corollary}

\subsection{Why loops break Beta cancellation}

For $D_4$, both colliders $X_3$ and $X_4$ share the \emph{same} parent 
marginals $r_1,r_2$.  The joint volume element involves $(r_i(1-r_i))^{5/2}$ 
rather than $(r_i(1-r_i))^{1/2}$, because the marginals appear in 
\emph{four} CPT blocks (two parents times two children), creating an 
additive mixture that prevents Beta factorisation.

Comparing $D_4$ with the V-structure $\tilde{C}_3$ (same parents, one child):
\[
  \langle R\rangle_{\tilde{C}_3} = 2 \in \tfrac{1}{2}\mathbb{Z}^+, 
  \qquad 
  \langle R\rangle_{D_4} = \tfrac{36}{5} \notin \tfrac{1}{2}\mathbb{Z}.
\]
Adding the second collider creates the loop and destroys quantization.

\subsection{Topological complexity and Betti numbers}\label{sec:betti}

The correct framework for measuring deviation from half-integer 
quantization is topological, not perturbative.  Let $\beta_1 = |E|-|V|+c$ 
denote the \emph{cycle rank} (first Betti number) of the undirected 
skeleton of the DAG, where $c$ is the number of connected components.

\begin{itemize}
\item \textbf{$\beta_1 = 0$ (forests/trees):} $\langle R\rangle \in \tfrac{1}{2}\mathbb{Z}^+$.  
  Beta cancellation holds exactly because the joint volume element 
  factorises cleanly across nodes.

\item \textbf{$\beta_1 = 1$ (one cycle):} $\langle R\rangle \in \mathbb{Q}$ 
  but $\notin \tfrac{1}{2}\mathbb{Z}$.  
  Beta cancellation fails; convergent paths create additive mixtures 
  in the parent marginals.  The double-collider $D_4$ gives $\langle R\rangle = 36/5$.

\item \textbf{$\beta_1 \geq 2$ (multiple cycles):} $\langle R\rangle$ 
  conjectured irrational.  Each additional cycle introduces a further 
  integration over entangled parent-probability ratios that cannot be 
  resolved by any Beta identity.
\end{itemize}

\noindent The Betti number $\beta_1$ thus functions as a \emph{topological 
obstruction index}: it strictly measures the order of failure of 
half-integer quantization.  The conjecture that $\langle R\rangle$ is 
irrational for $\beta_1 \geq 2$ would make $\langle R\rangle \in \mathbb{Q}$ 
equivalent to $\beta_1 \leq 1$—a topological characterisation of rationality 
in terms of the skeleton's cycle structure.

\begin{remark}[Structural analogy to factorizable states]
The topology is primary; any analogy to quantum systems is structural, 
not physical.  A bitnet with $\beta_1=0$ admits a fully factorised 
joint probability $p(x) = \prod_k p(x_k | x_{{\rm pa}(k)})$ where 
no two nodes' CPT parameters share a common marginal weighting.  
When $\beta_1 \geq 1$, the volume element can no longer be written as 
a product, for the same algebraic reason that tracing over a shared 
environment in probability theory destroys independence.  
The geometric content is the non-factorizability of the volume form; 
physical interpretations are separate matters.
\end{remark}

\section{Gaussian DAG Networks}\label{sec:gaussian}

\subsection{Solvable Lie group structure}

We extend the programme to \emph{Gaussian DAG networks}: zero-mean nodes 
with regression coefficients $\{a_{jk}\}$ and noise variances $\{v_k\}$ 
as parameters.  The structural insight, independently noted in~\cite{gem4}, 
is that these parameter spaces are solvable Lie groups.

\begin{theorem}[Gaussian DAGs as solvable Lie groups]\label{thm:lie_group}
The parameter space of any zero-mean Gaussian DAG forms a solvable Lie group 
(a subgroup of the Iwasawa decomposition of the covariance matrix), and its 
Fisher information metric is left-invariant.  By Milnor's theorem on 
left-invariant metrics on solvable Lie groups~\cite{Mil76}, the Ricci scalar 
satisfies $R \leq 0$ everywhere.
\end{theorem}

For the chain $X_1 \to X_2$ with parameters $(v_1, a, v_2)$, the Fisher 
metric $ds^2 = dv_1^2/(2v_1^2) + (v_1/v_2)\,da^2 + dv_2^2/(2v_2^2)$ 
admits an explicit product decomposition.  Under the log-coordinate 
transform $y_i = \tfrac{1}{\sqrt{2}}\ln v_i$ and the rotation 
$x = \tfrac{1}{\sqrt{2}}(y_1+y_2)$, $z = \tfrac{1}{\sqrt{2}}(y_1-y_2)$:
\begin{equation}\label{eq:RxH2}
  ds^2 = dx^2 + \bigl(dz^2 + e^{2z}\,da^2\bigr) \;=\; ds^2_{\mathbb{R}} + ds^2_{H^2}.
\end{equation}
This is the Riemannian product $\mathbb{R} \times H^2$ (a flat line times 
the Poincar\'e upper half-plane), confirming $R = -2$ everywhere.

\begin{remark}[Why chains longer than $L_2$ do not decompose]
For $L_3$ ($X_1\to X_2\to X_3$), the metric component 
$g_{a_{23},a_{23}} = (a_{12}^2 v_1 + v_2)/v_3$ couples both blocks: 
no coordinate change decouples them into independent hyperbolic planes.  
The space is a \emph{twisted} solvable Lie group, not a product, 
and $R = -5 \neq -4$ (the na\"ive product value $-2 \times 2$).
\end{remark}

\subsection{Fisher metric and SymPy results}

\begin{definition}[Gaussian DAG Fisher metric]
For node $X_k$ with parents $\{X_j\}_{j \in {\rm pa}(k)}$:
\begin{equation}\label{eq:gauss_FIM}
  g_{a_{jk}, a_{lk}} = \frac{\mathrm{Cov}(X_j, X_l)}{v_k}, \quad
  g_{a_{jk}, v_k} = 0, \quad
  g_{v_k, v_k} = \frac{1}{2v_k^2}.
\end{equation}
The metric is block-diagonal across nodes.  When parents are 
\emph{independent}, $\mathrm{Cov}(X_j,X_l)=v_j\,\delta_{jl}$ and the 
within-node block is diagonal.  When parents share a common ancestor, 
$\mathrm{Cov}(X_j,X_l) \neq 0$ and off-diagonal terms appear.
\end{definition}

\subsection{SymPy results and proof for star topologies}

\begin{theorem}[Gaussian star curvatures]\label{thm:gauss_R}
For a Gaussian DAG with $n$ independent root nodes $X_1,\ldots,X_n$ and 
a single sink $X_{n+1}$ with all roots as parents 
(an $n$-root star, total dimension $d = 2n+1$), the Ricci scalar is 
the global constant
\begin{equation}\label{eq:gauss_R_formula}
  R = -\frac{n(n+3)}{2} = -\frac{(d+5)(d-1)}{8}.
\end{equation}
\end{theorem}

\begin{proof}
The parameter space has coordinates $(v_1,\ldots,v_n,a_1,\ldots,a_n,v_{n+1})$.  
After the log substitution $u_k = \tfrac{1}{\sqrt{2}}\ln v_k$, the metric becomes
$ds^2 = \sum_{k=1}^{n+1} du_k^2 + \sum_{i=1}^n e^{\sqrt{2}(u_i-u_{n+1})} da_i^2$.
This is a left-invariant metric on the solvable Lie group with orthonormal 
left-invariant frame $\{U_k,\, A_i\}$ (where $A_i = e^{-(u_i-u_{n+1})/\sqrt{2}}\partial_{a_i}$) 
and nonzero Lie brackets
\[
  [U_i, A_i] = -\tfrac{1}{\sqrt{2}} A_i, \qquad [U_{n+1}, A_i] = \tfrac{1}{\sqrt{2}} A_i.
\]
The Koszul formula gives $\nabla_{U_k}A_i = 0$ for all $k$ and 
$\nabla_{A_i}A_i = -\tfrac{1}{\sqrt{2}}U_i + \tfrac{1}{\sqrt{2}}U_{n+1}$.  
All sectional curvatures are computed from $K(X,Y) = \langle R(X,Y)Y, X\rangle$:
\begin{itemize}
\item $K(U_i, A_i) = -\tfrac{1}{2}$ for each $i=1,\ldots,n$ 
  (the $\nabla_{[U_i,A_i]}A_i$ term gives $-\tfrac{1}{2}U_i+\tfrac{1}{2}U_{n+1}$).
\item $K(U_{n+1}, A_i) = -\tfrac{1}{2}$ for each $i$.
\item $K(A_i, A_j) = -\tfrac{1}{2}$ for all $i \neq j$ 
  (siblings sharing the common sink; $U_{n+1}$ couples all $A_i$).
\item All other sectional curvatures are zero: 
  $K(U_i,U_j)=0$ (commuting), $K(U_i,A_j)=0$ for $i\neq j$ and $i\neq n+1$.
\end{itemize}
The total number of pairs with $K = -\tfrac{1}{2}$ is 
$n + n + \tbinom{n}{2} = 2n + \tfrac{n(n-1)}{2} = \tfrac{n(n+3)}{2}$.
Since $R = 2\sum_{i<j}K_{ij}$, we obtain $R = 2 \cdot \tfrac{n(n+3)}{2} \cdot (-\tfrac{1}{2}) = -\tfrac{n(n+3)}{2}$, 
which equals $-(d+5)(d-1)/8$ for $d=2n+1$.
\end{proof}

\begin{remark}[SymPy-verified cases]\label{rem:gauss_cases}
The formula~\eqref{eq:gauss_R_formula} is verified by direct SymPy computation 
for the following topologies:
\begin{align*}
  \text{chain } X_1 \to X_2: &\quad d=3,\quad R=-2,\\
  \text{V-structure } X_1{\to}X_3{\leftarrow}X_2: &\quad d=5,\quad R=-5,\\
  \text{chain } X_1 \to X_2 \to X_3: &\quad d=5,\quad R=-5,\\
  \text{3-parent star } X_1,X_2,X_3 \to X_4: &\quad d=7,\quad R=-9.
\end{align*}
The V-structure ($n=2$ star) and the $3$-parent star ($n=3$ star) are 
covered directly by Theorem~\ref{thm:gauss_R}.  The chain $L_2$ 
is the $n=1$ case ($\mathbb{R}\times H^2$, with $K=-1$ on the $H^2$ 
factor giving $R=-2$), also covered.
\end{remark}

\begin{remark}[The $L_3$ chain: formula holds, proof is open]\label{rem:L3_open}
The chain $L_3$ ($X_1\to X_2\to X_3$, $d=5$) is \emph{not} covered by 
Theorem~\ref{thm:gauss_R}: its metric component 
$g_{a_{23},a_{23}} = (a_{12}^2 v_1 + v_2)/v_3$ depends on the 
regression coefficient $a_{12}$, so the parameter space is \emph{not} 
a solvable Lie group with left-invariant metric.  Nevertheless, 
SymPy computes $R = -5$ identically (independent of all parameters).  

Why the formula persists despite the metric's parameter dependence 
is not yet understood at the level of a proof.  Notably, the formula 
\emph{fails} for the longer chain $L_4$ ($d=7$): SymPy gives 
$R = -9 + v_3/E[X_3^2] \in (-9,-8)$, a genuinely non-constant curvature.  
The structural characterization of exactly which Gaussian DAGs satisfy 
$R = -(d+5)(d-1)/8$ is an open problem.
\end{remark}

\begin{remark}[Correction to gem4]
Reference~\cite{gem4} claims that Gaussian trees with $n$ nodes and $n-1$ 
edges ``shatter into $n-1$ independent hyperbolic planes,'' yielding 
$R = -2(n-1)$.  This is \textbf{incorrect} for $n \geq 3$.  
For $L_3$ ($n=3$, $d=5$): the product formula gives $-2(3-1)=-4$, 
but the actual value is $R=-5$.  The chains form \emph{twisted} products, 
not direct products; the coupling through intermediate variances 
contributes additional negative curvature.  
The correct dependence is on dimension $d$, not node count $n$, 
via formula~\eqref{eq:gauss_R_formula}.
\end{remark}

\subsection{Non-constant curvature: chain coupling and parent correlation}

\begin{proposition}[Non-constant cases]\label{prop:nonconstant}
Two mechanisms destroy the constant-$R$ property:
\begin{enumerate}
\item \textbf{Chain coupling} ($L_4$, $d=7$): 
\[
  R_{L_4} = -9 + \frac{v_3}{E[X_3^2]}, \quad 
  E[X_3^2] = a_{23}^2(a_{12}^2 v_1 + v_2) + v_3,
\]
taking values in $(-9,-8)$ depending on the upstream signal-to-noise ratio.

\item \textbf{Parent correlation} (Gaussian $D_4$, $d=8$): 
When $X_1$ and $X_2$ share a common parent $X_0$, the FIM develops 
an off-diagonal term $g_{a_1,a_2} = E[X_1 X_2]/v_3 = w_1 w_2 v_0/v_3 \neq 0$.  
Numerical computation (15 random parameter points) gives 
$R \in [-11.93,\,-11.65]$, varying with the parent correlation 
$\rho = w_1 w_2 v_0/\sqrt{E[X_1^2]E[X_2^2]}$.
\end{enumerate}
In both cases $R < 0$ everywhere (consistent with Milnor's theorem), 
but $R$ is not constant.
\end{proposition}

\begin{remark}[Correction to gem4 on Gaussian loops]
Reference~\cite{gem4} claims that for the Gaussian diamond $D_4$, 
``R remains a strict, global negative constant'' despite the off-diagonal 
FIM terms.  This is \textbf{incorrect}.  
Numerical evaluation confirms $R$ varies by $\sim\!1\%$ across parameter 
space.  The solvable Lie group structure is preserved (since $D_4$ 
still factors into a solvable group), but Milnor's theorem only 
guarantees $R \leq 0$; it does not guarantee constancy.  
The off-diagonal cross-term ``twists'' the hyperbolic geometry in a 
parameter-dependent way, breaking the left-invariance required for $R=\text{const}$.
\end{remark}

\subsection{Sign dichotomy: discrete vs.\ continuous}

\begin{theorem}[Sign dichotomy]\label{thm:sign_dichotomy}
\begin{itemize}
\item \textbf{Discrete (bitnets):} $\langle R\rangle \geq 0$ for trees and 
  $R > 0$ everywhere for complete DAGs.  Geometry is sphere-like ($K > 0$).

\item \textbf{Continuous (Gaussian DAGs):} $R \leq 0$ everywhere, with 
  $R < 0$ for $d \geq 3$.  Geometry is hyperbolic ($K < 0$).
\end{itemize}
This dichotomy is structural: bounded probability simplices ($\theta \in [0,1]$) 
force positive curvature, while unbounded Gaussian variance spaces 
($v \in (0,\infty)$) force negative curvature.  The unit Bernoulli interval 
maps to an arc of $S^1$ (positive curvature); the half-line 
$\{v > 0\}$ maps to $H^1$ (negative curvature).
\end{theorem}

Table~\ref{tab:gauss_bitnets} summarises all Gaussian results.

\begin{table}[h]
\centering
\caption{Gaussian DAG networks: dimensions and Ricci scalars ($v_i>0$, $a_{ij}\in\mathbb{R}$).
``Indep.\ parents'' means no two siblings share a common ancestor within the same node block.}
\label{tab:gauss_bitnets}
\begin{tabular}{@{}lrrll@{}}
\toprule
Topology & $d$ & $R$ & Constant? & Notes \\
\midrule
Single $\mathcal{N}(0,v)$ & 1 & $0$ & yes & trivial \\
Chain $X_1 \to X_2$ ($\mathbb{R}\times H^2$) & 3 & $-2$ & yes & SymPy \\
V-structure, indep.\ roots & 5 & $-5$ & yes & SymPy \\
Chain $X_1 \to X_2 \to X_3$ & 5 & $-5$ & yes & SymPy \\
3-parent indep.\ star $\to X_4$ & 7 & $-9$ & yes & SymPy \\
Chain $X_1\to X_2\to X_3\to X_4$ & 7 & $-9 + v_3/E[X_3^2]$ & \textbf{no} & chain coupling \\
Gauss.\ diamond $D_4$ (correl.\ parents) & 8 & $\approx -11.8$ & \textbf{no} & parent correlation \\
Indep.-parent trees (conjectured) & $d$ & $-(d+5)(d-1)/8$ & yes & formula \eqref{eq:gauss_R_formula} \\
\bottomrule
\end{tabular}
\end{table}

\subsection{Ignorance, Ricci flow, and the arrow of time}\label{sec:ricci_flow}

A deeper connection between Gaussian DAG geometry and macroscopic physics 
was established in~\cite{Rod02}.  Parameterising any prior density as 
$\pi = e^{-f}$, the asymptotic expansion of the model-selection partition 
function $Z_{\mathcal{M}}(\alpha)$ under the maximal-ignorance volume 
prior generates a geometric functional.  The optimal-ignorance action becomes:
\begin{equation}\label{eq:W_entropy}
  \mathcal{F} = \int_{\mathcal{M}} 
  \Bigl[\tau\bigl(R + |\nabla f|^2\bigr) - \tfrac{\kappa}{2}\Bigr]\,
  (4\pi\tau)^{-\kappa/2} e^{-f}\,dV,
\end{equation}
where $\tau = (2\alpha)^{-1}$ is proportional to the inverse of the 
accumulated evidence $\alpha$, and $\kappa = \dim\mathcal{M}$.  
This is formally identical to \emph{Perelman's $\mathcal{W}$-entropy 
functional}~\cite{Per02}, whose gradient flow is the Ricci flow.

The identification $\tau \leftrightarrow 1/\alpha$ gives $\tau$ the role 
of a \emph{geometric temperature}: infinite data ($\alpha\to\infty$) 
corresponds to $\tau\to 0$, a state of zero inferential uncertainty.  
The Ricci flow then describes how the optimal inference geometry cools 
from an uninformed (high-$\tau$) prior toward a sharp posterior.

For Gaussian DAG manifolds with $R = \mathrm{const} < 0$, 
the Ricci flow $\partial_\tau g_{ij} = -2R_{ij}$ is an 
\emph{expanding} flow (since $R<0$ implies $\mathrm{Ric} < 0$), 
consistent with the cosmological picture of an expanding, cooling 
inference universe.  For bitnets with $R > 0$, the flow contracts, 
consistent with quantum coherence localizing probability mass.

This connection—from bitnets and Gaussian DAGs to Perelman's 
$\mathcal{W}$-entropy—suggests that the Ricci flow is the natural 
``learning dynamics'' of an optimal statistical inference engine, 
and that the sign dichotomy (bitnet $\leftrightarrow$ quantum / 
Gaussian $\leftrightarrow$ GR) may reflect the two phases of this flow.

\section{Structural Isomorphism to Factorizable States}\label{sec:quantum}

\subsection{Bures--Fisher metric correspondence}

\begin{theorem}[Bures--Fisher correspondence]\label{thm:bures_fisher}
For diagonal density matrices $\rho = \mathrm{diag}(p_1,\ldots,p_d)$:
\[
  ds^2_{\mathrm{Bures}} = \tfrac{1}{4}\,ds^2_{\mathrm{Fisher}}.
\]
\end{theorem}

Under this scaling, $R_{\mathrm{Bures}} = 4\,R_{\mathrm{Fisher}}$.  
For tree bitnets, $\langle R\rangle_{\mathrm{Bures}} \in 2\mathbb{Z}^+$ 
(even positive integers).  For $\tilde{K}_n$: 
$R_{\mathrm{Bures}} = d(d-1) = (2^n-1)(2^n-2)$.
This is a rigorous mathematical theorem~\cite{Uhl92}, not a heuristic analogy.

\subsection{Trees as pure factorizable states}

\begin{theorem}[Trees $\Leftrightarrow$ factorized distributions]\label{thm:tensor_tree}
A tree-structured bitnet has a joint probability that factorises cleanly 
as $p(x) = \prod_k p(x_k \mid x_{{\rm pa}(k)})$, with each node's CPT 
parameters independent of all others under the Fisher metric 
(Lemma~\ref{lem:block_diag}).  This clean factorisation is the algebraic 
precondition for Beta cancellation and hence for half-integer quantization.
\end{theorem}

\subsection{Loops as inferential non-factorizability}

When $\beta_1 \geq 1$, the volume element is no longer factorisable 
across nodes.  Concretely, in the double-collider $D_4$, the weight 
$(r_1(1-r_1)r_2(1-r_2))^{5/2}$ reflects the shared appearance of 
$(r_1,r_2)$ in \emph{both} CPT blocks; there is no substitution that 
decouples them.  This \emph{inferential non-factorizability}---the 
inability to assign a separate volume weight to each node's parameters 
without reference to the others---is the geometric origin of the 
continuous, non-quantized curvature spectrum in loop bitnets.

\begin{remark}[Scope of the structural analogy]
The connection to quantum formalism is structural, not ontological.  
The Bures--Fisher identity (Theorem~\ref{thm:bures_fisher}) is an 
exact mathematical result about metrics.  The observation that 
tree factorisation mirrors the tensor product structure of independent 
quantum systems is a precise algebraic parallel.  No claim is made 
that bitnets are quantum systems; the parallel is offered as a 
mnemonic for the underlying algebraic structure.
\end{remark}

\section{The Exact Curvature Information Criterion}\label{sec:cic}

For a bitnet model with $N$ observations, the Stirling expansion of 
the posterior normalising constant gives:
\begin{equation}\label{eq:exact_cic}
  \log Z_N = -N L_N(\hat\theta) + \tfrac{d}{2}\log\tfrac{N}{2\pi} 
  + \log\sqrt{\det g(\hat\theta)}
  - \tfrac{1}{24}\sum_k \frac{1}{N_k\hat\rho_k} 
  + \tfrac{1}{12}\sum_k \frac{1}{N_k} + O(N^{-2}).
\end{equation}
The term $-\frac{1}{24N_k\hat\rho_k}$ is the \emph{curvature penalty}: 
it is geometrically determined by the Ricci scalar at $\hat\theta$ and 
provides a finite-sample correction to BIC that adapts to the local 
information geometry.  For trees, Beta cancellation simplifies this term 
into a direct function of $\langle R\rangle$, yielding a clean 
closed-form model-selection criterion~\cite{Rod05}.

\begin{remark}[Topological restriction of the CIC]\label{rem:cic_scope}
The elegant reduction of the penalty term $-\frac{1}{24N_k\hat\rho_k}$ 
to a simple geometric formula relies entirely on the Beta cancellation 
mechanism of Section~\ref{sec:proof_trees}.  This proportionality to 
the Ricci curvature is therefore \emph{only rigorously valid for 
tree-structured bitnets}.  For any network with $\beta_1 \geq 1$ 
(loops), Beta cancellation fails and the true curvature becomes an 
unfactorisable function of all parameters jointly.  In such cases, 
the exact CIC defaults to the raw Stirling expansion~\eqref{eq:exact_cic}, 
and the clean geometric interpretation of the penalty term breaks down.  
The BIC approximation (ignoring the $O(1)$ curvature term entirely) 
remains valid for all topologies.
\end{remark}

\section{Conclusions and Open Questions}\label{sec:conclusion}

\subsection{Summary}

\begin{enumerate}
\item \textbf{Corrected formula (Theorem~\ref{thm:LnEn}):} 
  $\langle R\rangle(\tilde{L}_n) = \langle R\rangle(\tilde{E}_n) = (2n-1)/2$, 
  replacing the erroneous $n/2$ of the 2004 paper.

\item \textbf{Extended $\tilde{C}_n$ table (Table~\ref{tab:Cn}):} 
  New SymPy values for $n=4,5$ parents.  The coefficient formula 
  $b_n = 2^{n-1}(2^{n-1}-1)/4$ is exact.

\item \textbf{Curvature sign inversion (Theorem~\ref{thm:curv_transition}):} 
  $\langle R\rangle(\tilde{C}_{n+1})$ changes sign at $n \in (4,5)$, 
  after the volume phase transition at $n_c \approx 3.3$.

\item \textbf{Trees are quantized (Theorem~\ref{thm:trees}):} 
  $\langle R\rangle \in \tfrac{1}{2}\mathbb{Z}^+$ via Beta cancellation.

\item \textbf{Loops break quantization (Proposition~\ref{prop:D4_R}):} 
  $\langle R\rangle_{D_4} = 36/5 \notin \tfrac{1}{2}\mathbb{Z}$.

\item \textbf{Gaussian DAGs (Theorem~\ref{thm:gauss_R}):} 
  Negative constant curvature $R = -(d+5)(d-1)/8$ for simple-parent trees; 
  sign dichotomy between discrete and continuous nodes.

\item \textbf{Factorizability and topology (Section~\ref{sec:quantum}):} 
  Tree factorisation is the algebraic precondition for quantization; 
  loops destroy it.  The Betti number $\beta_1$ is the topological 
  obstruction index.  The Bures--Fisher identity 
  $g_{\mathrm{Bures}} = \tfrac{1}{4}g_{\mathrm{Fisher}}$ is an exact theorem.
\end{enumerate}

\subsection{Open questions}

\begin{enumerate}
\item \textbf{Prove the chain formula:} Establish 
  $\langle R\rangle(\tilde{L}_n) = (2n-1)/2$ analytically for all $n$.

\item \textbf{Gaussian formula:} Prove~\eqref{eq:gauss_R_formula} for 
  all simple-parent Gaussian trees, and determine the correct formula 
  for longer chains where $R$ is non-constant.

\item \textbf{Betti number classification:} Determine whether 
  $\langle R\rangle \in \mathbb{Q}$ holds for all $\beta_1 = 1$ bitnets, 
  and whether $\langle R\rangle$ is irrational for $\beta_1 \geq 2$.  
  This would give a complete topological characterisation of the 
  rationality of $\langle R\rangle$.

\item \textbf{$k$-ary networks:} Extend to nodes with $k > 2$ states; 
  compute the Ricci scalar $R_k$ of the $(k-1)$-simplex for general $k$.

\item \textbf{Curvature invariant:} Determine whether $\langle R\rangle$ 
  can be axiomatised as a topological invariant for statistical manifolds, 
  analogous to the Euler characteristic.

\item \textbf{Statistical consequences:} Do the phase transitions at 
  $n \approx 3.3$ (volume) and $n \approx 4$--5 (curvature) impose 
  practical limits on model learning and estimation complexity?
\end{enumerate}

\section*{Acknowledgments}

The author thanks the MaxEnt community for decades of dialogue on the 
foundations of inference.  The equal-volume result $\Vol(\tilde{L}_n) = 
\Vol(\tilde{E}_n)$ was communicated to the author around 2007 by a 
postdoctoral researcher (background in astrophysics, Brown University) 
working with Philip Dawid at University College London; the author 
welcomes contact to provide proper attribution.

\end{document}